\documentclass[11pt]{article}
\usepackage[left=1in, top=1in, right=1in, bottom=1in]{geometry}
\usepackage{indentfirst}

\usepackage{latexsym}
\usepackage[linesnumbered,ruled,vlined]{algorithm2e}
\usepackage{amsmath, amsthm, amsfonts}
\usepackage{listings}
\usepackage{mathtools}
\usepackage{hyperref}
\usepackage{url}

\usepackage{tikz}
\usetikzlibrary{shapes,backgrounds,calc}
\usepackage{subcaption}

\newtheorem{theorem}{Theorem}

\theoremstyle{definition}
\newtheorem{definition}{Definition}

\newcommand{\pr}{\textsc{PageRank}} 
\newcommand{\fpr}{\textsc{Diverse Centrality}}
\newcommand{\prs}{\textsc{PR}} 
\newcommand{\fprs}{\textsc{DC}} 

\newcommand{\dbc}{\textsc{Diverse Betweenness Centrality}}
\newcommand{\dbcs}{\textsc{DBC}}

\newcommand{\bc}{\textsc{Betweenness Centrality}}
\newcommand{\bcs}{\textsc{BC}}

\newcommand{\rwnode}{\textsc{Reweight Node Bias}}
\newcommand{\rwneighbor}{\textsc{Reweight Neighbor Bias}}

\newcommand{\fr}{\textsc{Fully Random}}
\newcommand{\pa}{\textsc{Preferential Attachment}}
\newcommand{\ba}{\textsc{Polarity Attachment}}
\newcommand{\clp}{\textsc{Change Local Polarity}}
\newcommand{\cnp}{\textsc{Change Neighborhood Polarity}}
\newcommand{\cc}{\textsc{9 Clusters}}

\makeatletter
\tikzset{circle split part fill/.style  args={#1,#2}{%
 alias=tmp@name, 
  postaction={%
    insert path={
     \pgfextra{%
     \pgfpointdiff{\pgfpointanchor{\pgf@node@name}{center}}%
                  {\pgfpointanchor{\pgf@node@name}{east}}%
     \pgfmathsetmacro\insiderad{\pgf@x}
      \fill[#1] (\pgf@node@name.base) ([xshift=-\pgflinewidth]\pgf@node@name.east) arc
                          (0:180:\insiderad-\pgflinewidth)--cycle;
      \fill[#2] (\pgf@node@name.base) ([xshift=\pgflinewidth]\pgf@node@name.west)  arc
                           (180:360:\insiderad-\pgflinewidth)--cycle;            
         }}}}}  
 \makeatother  

\title{Centrality with Diversity}

\date{\today}

\makeatletter
\def\@fnsymbol#1{\ensuremath{\ifcase#1\or \dagger\or \ddagger\or
\mathsection\or \mathparagraph\or \|\or **\or \dagger\dagger
\or \ddagger\ddagger \else\@ctrerr\fi}}
\makeatother

\newcommand*\samethanks[1][\value{footnote}]{\footnotemark[#1]}

\author{
Liang Lyu\thanks{Duke University. \tt{lianglyu1998@gmail.com, \{btfain, kamesh, knwang\}@cs.duke.edu}.}
\and
Brandon Fain\samethanks[1]
\and
Kamesh Munagala\samethanks[1]
\and
Kangning Wang\samethanks[1]
}

\begin{document}

\maketitle

\begin{abstract}
    Graph centrality measures use the structure of a network to quantify central or ``important'' nodes, with applications in web search, social media analysis, and graphical data mining generally. Traditional centrality measures such as the well known \pr{} interpret a directed edge as a vote in favor of the importance of the linked node. We study the case where nodes may belong to diverse communities or interests and investigate centrality measures that can identify nodes that are simultaneously important to many such diverse communities. We propose a family of diverse centrality measures formed as fixed point solutions to a generalized nonlinear eigenvalue problem. Our measure can be efficiently computed on large graphs by iterated best response and we study its normative properties on both random graph models and real-world data. We find that we are consistently and efficiently able to identify the most important diverse nodes of a graph, that is, those that are simultaneously central to multiple communities.
\end{abstract}

\section{Introduction}
\label{sec:motivation}
A fundamental question in network analysis concerns \emph{graph centrality}: Which nodes are ``important'' to understanding the overall network structure? Particular measures of centrality vary from node degrees to citation indices in publication networks~\cite{Hirsch05} to the celebrated \pr{} originally developed for web search~\cite{Page99}. In such work, it is typically assumed that the graph is composed of a single homogeneous community of nodes. We explore the question of centrality in a context of graphs with many heterogeneous communities and propose to answer the question:
\begin{quote}
Which nodes are simultaneously ``important'' to multiple diverse communities within the graph?
\end{quote}

As a concrete example, consider the Twitter graph of members of the United States Congress (where the members are nodes and directed edges represent following relationships on Twitter). Traditional graph centrality measures attempt to quantify the importance of particular members within the social network of the Congress as a whole. It is well known that the United States Congress is (roughly) composed of two communities: Democrats (or liberals) and Republicans (or conservatives). We are interested in the following question: Who are the most important members in the Congress, both in terms of network structure and bipartisanship? 
We address this particular example in our real-world experiments in Section~\ref{sec:real}. As another example, suppose we study publication networks composed of several sub-disciplines and are interested in identifying papers or authors who are most important to several sub-disciplines simultaneously or who serve as interdisciplinary bridges between publication communities.

To address these questions, we consider a model where in addition to a directed graph, we have a score vector for each node corresponding to community affiliation. These score vectors are real-valued and denote the degree of affiliation with different communities. We say that a node is \textit{polarized} or \textit{biased} if it is mostly affiliated with a particular community, and \textit{balanced} if it is comparably affiliated with all communities. We begin with the same basic intuition as other classic measures of graph centrality, that a directed edge from one node to another constitutes a ``citation'' or ``vote in favor'' of the target node. In the standard \pr{} model, the centrality is a fixed point solution that sets the measure of a node $u$ to the average of the measures of the nodes that cite $u$. Thus, a node could have a high \pr{} even if it is only cited by members of its own community. On the other hand, in our model, a node should only receive a high \fpr{} if it is cited by nodes of high \fpr{} score from multiple communities. Informally, the \fpr{} of a node $u$ therefore is the \emph{minimum} over communities of the average \fpr{} of the nodes of that community that cite $u$. We present the formal and more general definition in Section~\ref{sec:model}.

\subsection{Contributions and Outline} 
Our first contribution in Section~\ref{sec:model} is the proposal of a general family of centrality measures called \fpr{} that capture the idea of centrality with diversity in graphs with community structure. \fpr{} is the fixed point solution to a generalized nonlinear eigenvalue problem; we first show that such a solution is guaranteed to exist and provide an iterative algorithm for computing a solution. 

In Section~\ref{sec:experiments} we define our simulation setup for studying \fpr{} on random graphs motivated by real-world networks, including Erd\H{o}s-R\`{e}nyi graphs~\cite{erdos59a}, preferential attachment graphs~\cite{ba}, and a polarity attachment variant inspired by affiliation networks~\cite{Lattanzi} that we introduce. In Section~\ref{subsec:iterations}, we show empirically that \fpr{} can be computed efficiently via iterated best response and that it converges to a unique equilibrium.

In Section~\ref{subsec:specific}, we present three criteria that a \fpr{} measure should satisfy. (1) Centrality: The measure should be greater for nodes with greater \pr{}. (2) Local Polarity: The measure should be higher for balanced nodes than polarized nodes. (3) Neighborhood Polarity: The measure should be higher for nodes with diverse citations. We show by simulations on the random graph models that in contrast to simple methods of re-weighting \pr{}, \fpr{} simultaneously satisfies all of these properties. An added advantage of \fpr{} compared to simple re-weighting is that it propagates diverse influence in the sense that it amplifies scores of nodes that are pointed to by nodes that are also influential and balanced.



In Section~\ref{sec:real} we study \fpr{} on real-world graphs. We consider two graphs: a network of political blogs and the Twitter network of the United States Congress. Both examples have externally verifiable community structure. We show that in addition to the aforementioned properties, 
the \fpr{} model favors nodes that are not only 
influential (in terms of \pr{}), but that are also bipartisan in their neighborhood structure.
Interestingly, we show that for the United States Congress, the top-scoring senators found by our algorithm correlate well with the senators that join the most number of bipartisan bills, while not all these senators are as highly ranked by standard \pr{}. 

Given our empirical observation that \fpr{} finds central nodes that also serve as bridges between communities, in Section~\ref{sec:dbc}, we propose and investigate a natural generalization of \bc{}~\cite{betweenness}, which we call \dbc{}. 
Though \dbc{} might appear better suited to capture our desiderata, we show \fpr{} dominates \dbc{} in our experiments, and is indeed the preferred method for achieving centrality with diversity.




\subsection{Related Work}
The theory and practice of \pr{}-based notions of graph centrality are significant topics in the literature on network analysis and data mining~\cite{Page99,Haveliwala02,Jeh03,Kamvar03,Jiang04}.
Much of the literature is concerned with the efficient computation of \pr{} on extremely large networks typical of the web graph. By contrast, we are primarily interested in understanding an alternate notion of centrality with diversity in a graph with community structure. 

There is a long  line of work on algorithms for community detection in social networks~\cite{Comm1,Comm2,Comm3,Comm4,Kloumann17}. 
In contrast with this work, we are interested in graphs with known community structure or affiliation and how that influences the determination of a diverse centrality measure.

There is also a rich literature on modeling how communities and polarizing opinions arise in evolving networks~\cite{D74,FJ90,HK02,WDA+02}. Polarization of opinions has become an increasingly important societal topic, since it is largely a corrosive process that leads to the splintering of society and the formation of opinion bubbles. Our paper can be viewed as taking a step towards mitigating the effect of polarization. Our centrality measure suggests a way to calculate influence in a ``depolarized'' manner by prioritizing nodes that are simultaneously important and bridge different viewpoints.

The \fpr{} measure can be viewed as enforcing the concept of max-min (or Rawlsian) fairness on centrality measures. Related notions of fairness have been widely studied in resource allocation~\cite{varian,budish2,envyFreeUpTo1}, classification~\cite{fairLossMinimization,calibration,DworkIndividual,Angwin,KleinbergMR}, 
and so on. The main difference in our work is that we enforce such fairness locally at each neighborhood in the graph, and examine the effect of propagating such a notion through the network. We show that such a process leads to interesting insights even in real-world networks.

The influential work of Everett and Borgatti~\cite{EB10} (see also~\cite{KK10}) proposes a number of measures for the centrality of a given group of network nodes. These are defined by extending the notions of degree centrality and \bc{} to measure the extent to which group members serve as ``bridge nodes'' in the network. For instance, group betweenness is defined as the sum over all pairs of nodes of the fraction of the shortest paths between those nodes that use group members.  Our focus is on the case where there are many groups and each node's group membership is {\em fractional}. In Section~\ref{sec:real}, we propose a natural generalization of \bc{} to our setting; however, we show empirically that it is not as effective as \fpr{}. 

Finally, just as \pr{} is a special case of the eigenvalue problem, the \fpr{} measure is a special case of  {\em non-linear eigenvalue} problems, which have a long and rich history in functional analysis and numerical methods; see~\cite{NLE} for a survey. In contrast to the standard eigenvalue problem that admits to efficient computation, there are no general efficient algorithms known for the non-linear generalizations. Nevertheless, we find that our notion of \fpr{} is simple enough that algorithms similar to those for computing \pr{} are effective in practice.

\section{\fpr{}}
\label{sec:model}
In this section we define \fpr{} as a fixed point solution to a generalized nonlinear eigenvalue problem and prove that a solution always exists. We then present an iterative algorithm to find such an equilibrium.


Let $(V,E)$ be a directed graph with $n$ nodes $V=\{1,2,\ldots,n\}$. Let $d_i$ be the outdegree of node $i$. In addition, we assume there are $K$ communities within the network. For each node $i$, there is an {\em affiliation vector} $\mathbf{q}^{(i)} \in \mathbb{R}_+^K$ with $\left\Vert \mathbf{q}^{(i)} \right\Vert_1 = 1$ characterizing the communities it belongs to. 

To be concrete, consider again the example of the Twitter graph of members of the US Congress. There are $K = 2$ communities. $q^{(i)}_1$ and $q^{(i)}_2 = 1 - q^{(i)}_1$ describe to what extent node $i$ belongs to the Democratic and Republican communities. A node with $\left|q^{(i)}_1 - q^{(i)}_2 \right|$ close to $1$ would be considered highly polarized or partisan; a value close to 0 would indicate balance or bipartisanship. 

Now we can define \fpr{}. To get the intuition, consider a single node $i$. We calculate the centrality of $i$ according to each community in turn by taking the average centrality of nodes that cite $i$, weighted by their affiliation with that community, and adding a damping term based on $i$'s own affiliation with that community. Then we take the minimum (or more generally some concave function) of these scores over all of the communities so that $i$'s centrality is limited by the communities for which $i$ is less central. This defines a nonlinear iterative procedure, to which \fpr{} is the fixed point solution. 

\begin{definition}
    \label{def:fairPageRank}
    The \emph{\fpr{}} of a graph $(V, E)$ with damping factor $p \in (0, 1)$  is a vector $\mathbf{s} \in \mathbb{R}_{\geq 0}^n$ such that $\left\Vert \mathbf{s} \right\Vert_1 = 1$, and for all $i \in V$,
    \[
    \lambda s_i = f\left( \left(1-p\right) \frac{\mathbf{q}^{(i)}}{n} + p \sum_{j:(j,i)\in E} \frac{s_j}{d_j} \mathbf{q}^{(j)} \right)
    \]
    where $f:\mathbb{R}^K_{\geq 0} \rightarrow \mathbb{R}_{\geq 0}$ is a concave function and $\lambda \in \mathbb{R}_+$ is the normalization constant. We require $f$ to map any strictly positive vector to a positive number, i.e. $x_i > 0 \ \forall i \in [K]$ implies $f(\mathbf{x}) > 0$.
\end{definition}

Note that if $K=1$ or $f(\mathbf{x}) = \Vert x \Vert_1$, then the notion becomes the standard \pr{}.  We are interested in the case where $K \geq 2$ and $f$ is concave to prioritize diversity. For example, we desire $f\big((0.5, 0.5)\big)$ to be greater than $f\big((0.1, 0.9)\big)$ or $f\big((0.9, 0.1)\big)$. Some convenient choices of $f$ include the minimum function and the geometric mean. In our experiments, we focus on the minimum. The damping terms of $\left(1-p\right) \frac{\mathbf{q}^{(i)}}{n}$ are similar to those in standard \pr{} -- $\frac{1-p}{n}$ -- but multiplied by $\mathbf{q}^{(i)}$ to account for the local polarity of the node. 
We use a damping factor of $p = 0.85$ for \pr{} and \fpr{} in our experiments.

\subsection{Existence of Equilibrium}
Now we establish the existence of \fpr{}.
\begin{theorem}
The \fpr{} of a directed graph always exists, that is, there exists $\mathbf{s}$ that satisfies Definition~\ref{def:fairPageRank}.
\end{theorem}
\begin{proof}
We use Brouwer's fixed-point theorem, which states that any continuous mapping from a compact and convex set to itself must have a fixed point, to show an equilibrium of \fpr{} always exists. Define function $G$ that maps the simplex $[0, 1]^n \cap \left\{ \mathbf{s}:\left\Vert \mathbf{s} \right\Vert_1 = 1 \right\}$ to itself such that $
G_i(\mathbf{s}) = \frac{t_i}{\sum_{j=1}^n t_j}$,    
where
\[
t_i = f\left(\left(1-p\right) \frac{\mathbf{q}^{(i)}}{n} + p \sum_{j:(j,i)\in E} \frac{s_j}{d_j} \mathbf{q}^{(j)} \right), \ \forall i \in [n].
\]

Intuitively, the function $G$ is almost $f$ except that it also maps the output of $f$ back to the simplex proportionally.
Note that by Definition~\ref{def:fairPageRank}, $f$ is concave and hence continuous; therefore, $G$ is also continuous as it is the ratio between two continuous positive functions (each $t_j$ is positive so the denominator is positive). Furthermore, $G_i(\mathbf{s}) > 0$ for all $i \in [n]$. The function $G$ is thus a continuous mapping from the compact and convex set $[0, 1]^n \cap \left\{ \mathbf{s}:\left\Vert \mathbf{s} \right\Vert_1 = 1 \right\}$ to itself. By Brouwer's fixed-point theorem, $G$ has a fixed point $\mathbf{s}$, i.e. $G(\mathbf{s}) = \mathbf{s}$. This $\mathbf{s}$ satisfies our definition of \fpr{} that $s_i$ is proportional to $t_i$.
\end{proof}


\subsection{Algorithm for \fpr{}}
We present Algorithm~\ref{alg:my} to compute \fpr{}. It starts from a random weight vector $\mathbf{s}^{(0)}$ and iteratively updates itself using the formula in Definition~\ref{def:fairPageRank}. It terminates when an iteration does not change the weight vector much: When the $L^1$-distance between the weight vectors of two iterations is smaller than the given precision $\epsilon$.

\begin{algorithm}[htbp]
\SetKwInOut{Input}{Input}
\SetKwInOut{Output}{Output}
\SetKwRepeat{Do}{do}{while}
\Input{$(V, E)$, $p$, $f(\cdot)$, $\epsilon$, $\left\{\mathbf{q}^{(i)}\right\}_{i \in [n]}$}
\BlankLine
Choose any $\mathbf{s}^{(0)} \in \mathbb{R}_{\geq 0}^n$ with $\left\Vert \mathbf{s}^{(0)} \right\Vert_1 = 1$, and set $k \gets 0$ \;
\Do{$\left\Vert \mathbf{s}^{(k)} - \mathbf{s}^{(k - 1)} \right\Vert_1 > \epsilon$}
{
$k \gets k + 1$\;
$t_i \gets f\left(\left(1-p\right) \frac{\mathbf{q}^{(i)}}{n} + p \sum_{j:(j,i)\in E} \frac{s^{(k - 1)}_j}{d_j} \mathbf{q}^{(j)} \right), \ \forall i \in [n]$\;
$s_i^{(k)} \gets \frac{t_i}{\sum_{j = 1}^n t_j}, \ \forall i \in [n]$\;
}
\Return  $\mathbf{s}^{(k)}$; 
\caption{\fpr{}}
\label{alg:my}
\end{algorithm}

Experimentally, we find that on most random graphs the algorithm converges to an equilibrium in a small number of iterations and the equilibrium is most likely unique. The experiments are described in detail in Section~\ref{subsec:iterations}.



%

\section{Simulation Setup} 
\label{sec:experiments}
In this section, we describe the simulation setup that we will use to study the convergence behavior of the algorithm, as well as its properties relative to simpler baselines. We perform this study using a collection of well-motivated random graph models.

For simplicity, all experiments in subsequent sections focus on the special case of $K=2$. We will refer to the two communities as blue and red, with $b_i$ and $r_i$ scores to denote node $i$'s affiliation to each community. Formally, $b_i = q^{(i)}_1$ and $r_i = q^{(i)}_2$. We will sometimes refer to these values as ``polarities'' of a node; note that having either $r_i$ or $b_i$ close to 1 indicates polarization within that community. We use a damping factor of $p = 0.85$ for \pr{} and \fpr{} in all the experiments in Sections~\ref{subsec:iterations} and~\ref{sec:real}.  Further, the function $f$ in Definition~\ref{def:fairPageRank} is chosen to be the minimum.

\subsection{Random Graph Models} 
\label{subsec:simplegraphs}
We use three classes of randomly generated undirected graphs, along with their modifications, in our experiments. 

A \fr{} graph is generated using the Erdős–Rényi model~\cite{erdos59a}. The graph contains $n=1000$ nodes, and each node's $r_i$ value is drawn from a uniform distribution on $(0,1)$ with $b_i = 1 - r_i$. Each edge is included with probability $e=0.2$.
    
A \pa{} graph is generated by a preferential attachment mechanism such as the Barabási–Albert model~\cite{ba}. The graph begins with a clique of $m=20$ nodes. New nodes are then added one at a time until a total of $n=1000$ nodes, and each new node $i$ is randomly connected to $m$ existing nodes with probability proportional to their degrees. Formally, the $m$ nodes are drawn from the distribution $p_j = \frac{d_j}{\sum_{k:k<i} d_k} \ \forall j<i$. All nodes in the graph have random polarities (i.e. $r_i$ is drawn uniformly at random from $(0,1)$). This generative model simulates social networks without any consideration of polarized community formation. 
    
A \ba{} graph has nodes that are connected based on how similar their polarities are, resulting in polarized clusters. This is motivated by affiliation networks~\cite{Lattanzi}. The graph contains $n=1000$ nodes with random polarities, and each pair of nodes $i$ and $j$ is connected with a probability of $p_{ij}=0.5\left( r_i r_j + b_i b_j \right)$. Nodes with similar polarities are thus more likely to be connected. This generative model simulates the creation of polarized neighborhoods in the graph, such as what may occur when people are more likely to connect to others with similar opinions as their own. 
%

\section{Convergence and Uniqueness} 
\label{subsec:iterations}
In this section, we compute \fpr{} scores and \pr{} scores on the random graphs introduced in Section \ref{subsec:simplegraphs}.  We compare the rate of convergence of the two models in order to assess the efficiency of Algorithm~\ref{alg:my}. We also examine the uniqueness of \fpr{} scores found by Algorithm~\ref{alg:my}.

\subsection{Rate of Convergence}
Figure~\ref{fig:iterations} shows the distribution of the number of iterations of \pr{} and \fpr{} on \fr{}, \pa{} and \ba{} graphs, when $\epsilon = 10^{-10}$. Each run uses a different randomly generated graph, and for both models the importance vector is initialized to $s_i = \frac{1}{n}$ for $1\leq i\leq n$.

\begin{figure}
    \centering
    \includegraphics[width=0.85\columnwidth]{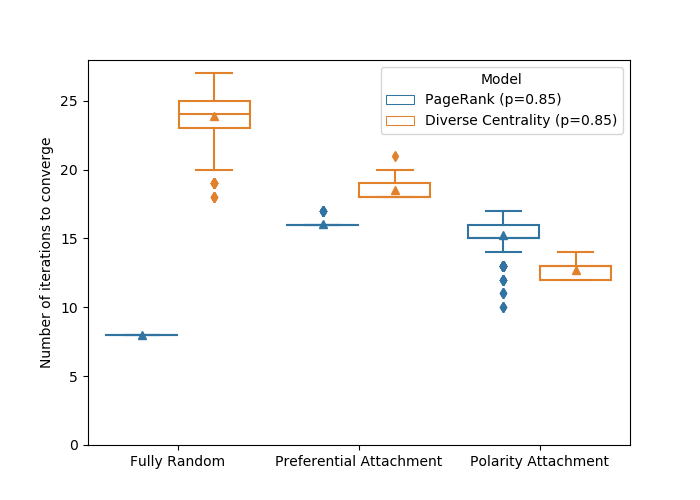}
    \caption{Number of iterations of \pr{} and \fpr{} on different graphs over 1000 runs.}
    \label{fig:iterations}
\end{figure}

The results suggest that not only does the \fpr{} model always converge to an equilibrium in practice, but its number of iterations is typically within a constant factor of \pr{}. The rate of convergence is the slowest on \fr{} graph with an average of $23.883$ iterations, or $2.99$ times of \pr{}. 


These experiments offer empirical evidence that the \fpr{} model converges in practice, and its rate of convergence is competitive compared to traditional centrality algorithms such as \pr{}. 

\subsection{Uniqueness of Equilibrium}
We also examine whether Algorithm~\ref{alg:my} always converges to a unique equilibrium independent of initialization. For this, we generate $600$ runs each for the $3$ graphs in Section~\ref{subsec:simplegraphs}. Each run uses a different random graph, and in each run Algorithm~\ref{alg:my} is executed twice, one with a constant initial vector $\mathbf{s}$ such that $s_i = \frac{1}{n} = 10^{-3} \ \forall 1\leq i\leq n$, and the other with a random $\mathbf{s}$. Each execution converges to a local equilibrium with $\epsilon = 10^{-10}$ ($L^1$-norm), so we calculate and analyze the differences of each node's scores between the two executions.

We also observe that the maximum difference of any node's score between the constant initialization and random initialization executions observed is $3.544 \times 10^{-11}$ (over all $1800$ runs on different graphs). This is insignificant compared to the average node score of $\frac{1}{n} = 10^{-3}$. 
Therefore, these experiments provide strong empirical evidence that Algorithm~\ref{alg:my} always converges to a unique global equilibrium, regardless of the initial vector $\mathbf{s}$. 

\section{Normative Properties}
\label{subsec:specific}
In this section, we compare the quality of the solution found by \fpr{} against simpler baselines. First, we propose three criteria that a centrality measure with diversity should satisfy:

\begin{description}
    \item[Centrality.] The measure should be higher for nodes that are more central according to standard graph centrality measures such as \pr{}.
    \item[Local Polarity.] The measure should be higher for nodes that belong to multiple communities in the sense of having balanced rather than polarized affiliation scores.
    \item[Neighborhood Polarity.] The measure should be higher for nodes that bridge communities in the sense that they are cited by multiple communities.
\end{description}

\subsection{Baseline Algorithms}    
We compare our \fpr{} model against three baseline centrality measures listed below. Note that by definition, the \pr{} model captures centrality, while \rwnode{} and \rwneighbor{} capture local polarity and neighborhood polarity respectively.

\begin{itemize}
    \item The  \pr{} model, with damping factor $p=0.85$.
    \item \rwnode{}: Re-weight the score vector of \pr{} proportional to $w_i = \min \{ r_i, b_i \}$. This means balanced nodes having greater $w_i$ will be considered more important.
    \item \rwneighbor{}: Re-weight the score vector of \pr{} proportional to $w_i = \min \left\{ \frac{R_i}{R_i+B_i}, \frac{B_i}{R_i+B_i} \right\}$, where
        \begin{align}
            R_i &= \sum_{\substack{j:(j,i)\in E\\j\neq i}} r_j + \sum_{\substack{j:(i,j)\in E\\j\neq i}} r_j,  \label{eq:R}\\
            B_i &= \sum_{\substack{j:(j,i)\in E\\j\neq i}} b_j + \sum_{\substack{j:(i,j)\in E\\j\neq i}} b_j.  \label{eq:B}
        \end{align}
    This means nodes in balanced neighborhoods which have a greater $w_i$ will be considered more important.
\end{itemize}

We show that \fpr{} captures all three criteria -- centrality, local polarity, and neighborhood polarity -- concurrently, while \pr{}, \rwnode{} and \rwneighbor{} cannot. Section~\ref{subsubsec:local} analyzes the effects of fixing neighborhood polarity and centrality while changing local polarity. Section~\ref{subsubsec:neighborhood} analyzes the effects of fixing local polarity and centrality while changing neighborhood polarity. Section~\ref{subsubsec:centrality} analyzes the effects of fixing local and neighborhood polarities while changing centrality. 

We further note that \fpr{} differs from the baselines in a more fundamental way. As mentioned before, it propagates influence in the sense that it amplifies the score of nodes that are themselves pointed to by nodes that are central and have balanced polarity. This feature is not shared by simple modifications of \pr{} based on node or neighborhood polarity.

\subsection{Local Polarity} 
\label{subsubsec:local}
Consider a bidirectional graph \clp{} with $2000$ nodes. The graph is first generated using the \fr{} model; then, $600$ nodes are chosen uniformly at random and divided into three sets:  
\begin{itemize}
    \item $V_1$ with $150$ nodes, each with polarity  $r_i = 0.99, b_i = 0.01$.
    \item $V_2$ with $300$ nodes, each with polarity  $r_i = b_i = 0.5$.
    \item $V_3$ with $150$ nodes, each with polarity $r_i = 0.01, b_i = 0.99$.
\end{itemize}

Since the \fr{} graph does not consider polarities during random generation, $V_1, V_2$ and $V_3$ will have close to balanced neighborhoods. The random selection of nodes ensures that centrality of these nodes will span the entire range of centrality values as measured by \pr{}.

\begin{figure}
    \centering
    \includegraphics[width=0.85\columnwidth]{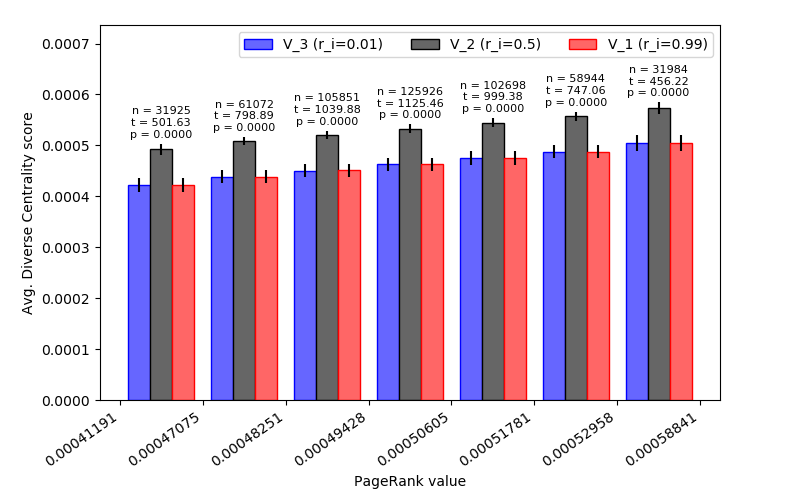}
    \caption{Average \fpr{} scores of nodes in \clp{} graph given range of \pr{} values and local polarity (864 runs).}
    \label{fig:change_local_fair}
\end{figure}

\begin{figure}
    \centering
    \includegraphics[width=0.85\columnwidth]{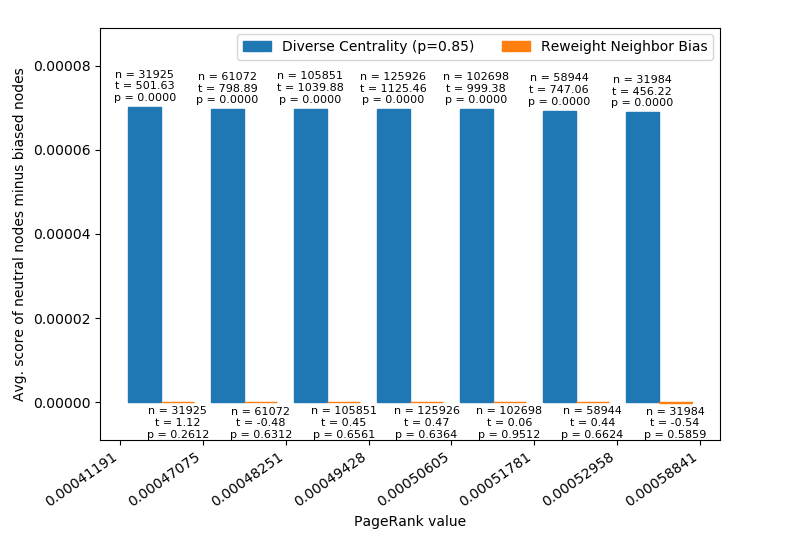}
    \caption{Average scores of balanced nodes minus polarized nodes in \clp{} graph, given range of \pr{} values, for \fpr{} and \rwneighbor{} models (864 runs).}
    \label{fig:change_local_compare}
\end{figure}

Figure~\ref{fig:change_local_fair} shows the average scores of nodes in $V_1$, $V_2$ and $V_3$ given by \fpr{}. The spectrum of \pr{} values is first divided into 7 buckets;\footnote{The 7 buckets are obtained by first dividing the entire range of \pr{} values evenly into 15 intervals, and then merging the first 5 and last 5 intervals which contain too few nodes individually.} nodes within each bucket are further grouped by their local polarity (i.e. whether they are in $V_1$, $V_2$ or $V_3$), and for each given bucket, the average \fpr{} scores of nodes from $V_1$, $V_2$ and $V_3$ are shown with the red, grey and blue bars respectively. Welch's $t$-tests are performed within each bucket, comparing all polarized nodes against all balanced nodes.

For both \fpr{} and \rwneighbor{}, we then compare the differences in scores between balanced nodes and polarized nodes in Figure~\ref{fig:change_local_compare}. Within each of the 7 buckets, for each model, the average scores of balanced nodes minus the average scores of polarized nodes is plotted, with $t$-tests performed similarly to Figure~\ref{fig:change_local_fair}.

Within each bucket of \pr{} values, nodes have approximately constant neighborhood polarity and centrality, but different local polarity. Both Figures~\ref{fig:change_local_fair} and~\ref{fig:change_local_compare} show that the \fpr{} model ranks balanced nodes in $V_2$ higher than polarized nodes in $V_1$ and $V_3$, and the difference is statistically significant in all 7 buckets. Conversely, the \rwneighbor{} model gives nodes in each bucket approximately the same scores regardless of their local polarity, with none of the differences being significant. 

This shows that our \fpr{} model reflects the differences in local polarities of nodes even when centrality and neighborhood polarity are held constant, while the \rwneighbor{} model fails to capture local polarity.

\subsection{Neighborhood Polarity} 
\label{subsubsec:neighborhood}


\begin{figure}[htbp]
    \centering
    \includegraphics[width=0.85\columnwidth]{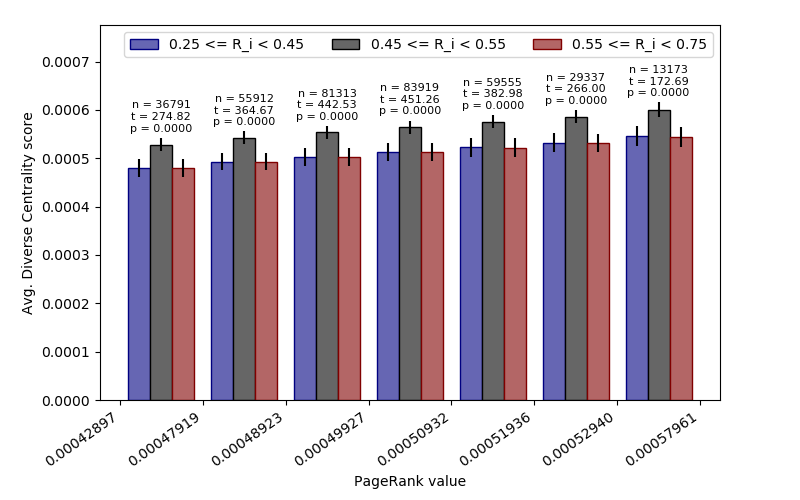}
    \caption{Average \fpr{} scores of nodes in \cnp{} graph given range of \pr{} values and neighborhood polarity (600 runs).}
    \label{fig:change_neighborhood_fair}
\end{figure}


Consider a bidirectional graph \cnp{} with $2000$ nodes. The graph is first generated using the \ba{} model to allow polarized clusters to form. Then, a set of $600$ nodes $V_0$ are chosen uniformly at random, and their polarities are changed to balanced with $r_i = b_i = 0.5$.

Although all nodes in $V_0$ have the same local polarity, their neighborhood polarity differs since neighborhood polarity in the \ba{} graph strongly correlates with their original polarities that were generated randomly.

Figure~\ref{fig:change_neighborhood_fair} shows the average scores of nodes in $V_0$ given by \fpr{}. Similar to Section~\ref{subsubsec:local}, the spectrum of \pr{} values is divided into 7 buckets. Nodes within each bucket are further grouped by their neighborhood polarity: node $i$ with $0.45 \leq R_i/(R_i+B_i) \leq 0.55$ is considered to be in a balanced neighborhood, where $R_i$ and $B_i$ are defined in Eq~(\ref{eq:R}) and~(\ref{eq:B}); otherwise, it is considered to be in a polarized neighborhood. Empirically, around $42\%$ of nodes in each bin are classified as in balanced neighborhoods.

For each given bucket, the average \fpr{} scores of nodes in red, balanced and blue neighborhoods (over 600 runs) are shown with the corresponding bars. Welch's $t$-tests are performed within each bucket, comparing all nodes in balanced neighborhoods against nodes in polarized neighborhoods.

Within each bucket of \pr{} values, nodes have the same local polarity and approximately constant centrality, but different neighborhood polarity. Figures~\ref{fig:change_neighborhood_fair} show that \fpr{} ranks nodes in balanced neighborhoods higher than nodes in polarized neighborhoods, with statistical significance in all buckets.


For both \fpr{} and \rwnode{} models, we compare the differences in scores between nodes in balanced and polarized neighborhoods, similarly to Figure~\ref{fig:change_local_compare}. Within all 7 buckets, the average \fpr{} scores of nodes in balanced neighborhoods is significantly higher than nodes in polarized neighborhoods, with differences ranging from $6.890 \times 10^{-5}$ to $7.011 \times 10^{-5}$. On the other hand, nodes in balanced neighborhoods typically have \textit{lower} \rwnode{} scores than those in polarized neighborhoods, with a much smaller magnitude of difference from $2.887 \times 10^{-9}$ to $1.192 \times 10^{-7}$.



These results show that our \fpr{} model captures the differences in neighborhood polarities of nodes even when centrality and local polarity are held constant, and gives balanced neighborhoods a higher score than polarized neighborhoods. \fpr{} thus captures our neighborhood polarity criterion, but \rwnode{} does not. 

\subsection{Centrality} 
\label{subsubsec:centrality}
Given that our \fpr{} algorithm has a similar formulation to standard \pr{}, one would naturally expect \fpr{} scores to correlate with \pr{}. Specifically, in the case where all nodes in the graph are balanced with $r_i = b_i = 0.5$, Definition~\ref{def:fairPageRank} reduces to \pr{} with damping factor $p$. However, we would like to analyze the relationship with centrality in more general, nontrivial cases where nodes may have varying polarities, and observe the effects of changing centrality on a subset of nodes whose local and neighborhood polarities are fixed.

\begin{figure}
    \begin{center}  
    \resizebox{0.8\columnwidth}{!}{%
        \begin{tikzpicture}
          [scale=.8,auto=left]
          \node[circle, draw=red, align=center, scale=0.7, minimum size=2cm] (A1) at (-3,3.2) {$50$ nodes};
          \node[circle, draw=blue, align=center, scale=0.7, minimum size=1.8cm] (A1blue) at (-3,3.2) {};
          \node[circle, draw=black, align=center, scale=0.7, minimum size=2cm] (B1) at (0,3.2) {$50$ nodes};
          \node[circle, draw=red, align=center, scale=0.7, minimum size=2cm] (C1) at (3,3.2) {$50$ nodes};
          \node[circle, draw=blue, align=center, scale=0.7, minimum size=1.8cm] (C1blue) at (3,3.2){};
          
          \node[circle, draw=red, align=center, minimum size=2cm] (A2) at (-4,0) {$150$ nodes};
          \node[circle, draw=blue, align=center,  minimum size=1.8cm] (A2blue) at (-4,0) {};
          \node[circle, draw=black, align=center, minimum size=2cm] (B2) at (0,0) {$150$ nodes};
          \node[circle, draw=red, align=center, minimum size=2cm] (C2) at (4,0) {$150$ nodes};
          \node[circle, draw=blue, align=center, minimum size=1.8cm] (C2blue) at (4,0) {};
          
          \node[circle, draw=red, align=center, scale=1.3, minimum size=2cm] (A3) at (-5,-4) {$450$ nodes};
          \node[circle, draw=blue, align=center, scale=1.3, minimum size=1.8cm] (A3blue) at (-5,-4) {};
          \node[circle, draw=black, align=center, scale=1.3, minimum size=2cm] (B3) at (0,-4) {$450$ nodes};
          \node[circle, draw=red, align=center, scale=1.3, minimum size=2cm] (C3) at (5,-4) {$450$ nodes};
          \node[circle, draw=blue, align=center, scale=1.3, minimum size=1.8cm] (C3blue) at (5,-4){};
          
          \node at (-4.5,3.2) {$A_1$};
          \node at (0.7,2.1) {$B_1$};
          \node at (4.5,3.2) {$C_1$};
          \node at (-5.8,0) {$A_2$};
          \node at (0.7,-1.5) {$B_2$};
          \node at (5.8,0) {$C_2$};
          \node at (-7.3,-4) {$A_3$};
          \node at (0.7,-6) {$B_3$};
          \node at (7.3,-4) {$C_3$};
          
          \draw[-] (A1) to (B1);
          \draw[-] (B1) to (C1);
          \draw[-] (A2) to (B2);
          \draw[-] (B2) to (C2);
          \draw[-] (A3) to (B3);
          \draw[-] (B3) to (C3);
          \draw[-] (B1) to (B2);
          \draw[-] (B2) to (B3);
        \end{tikzpicture}
    }
    \end{center}
    \caption{\cc{} graph.}
    \label{fig:9-clusters}
\end{figure}

\begin{figure}
    \centering
    \includegraphics[width=0.85\columnwidth]{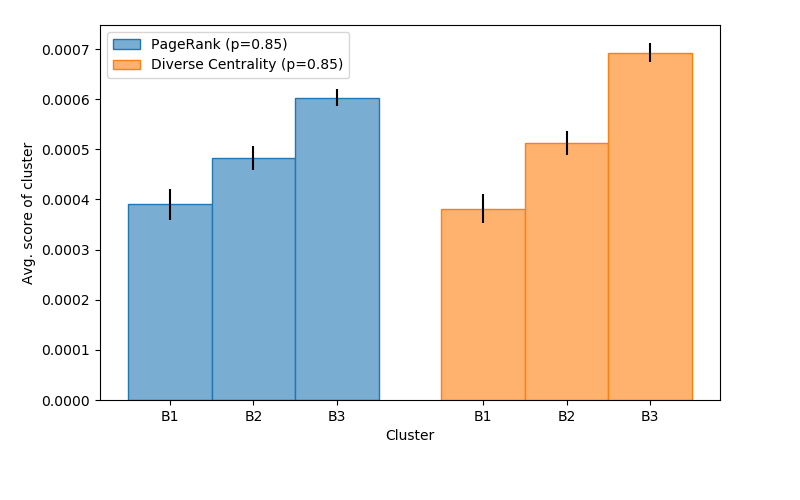}
    \caption{Average scores of $B_1$, $B_2$ and $B_3$ in \cc{} graph from \pr{} and \fpr{} (300 runs).}
    \label{fig:9_clusters}
\end{figure}

Consider a bidirectional graph \cc{} whose general structure is shown in Figure~\ref{fig:9-clusters}. The graph contains nine clusters: nodes in the central clusters $B_1$, $B_2$ and $B_3$ all have polarities $r_i = b_i = 0.5$, while nodes in other clusters have random polarities. Each pair of nodes within the same cluster have a probability of $0.5$ of being connected, while each pair of nodes across two clusters adjacent in Figure~\ref{fig:9-clusters} have a probability of $0.1$ of being connected.

All nodes in $B_1$, $B_2$ and $B_3$ are balanced, and have largely balanced neighborhoods due to randomness.\footnote{Empirically, over 300 runs, $R_i/(R_i+B_i)$ values of nodes in $B_1$, $B_2$ and $B_3$ have a mean of $0.499987$ and standard deviation of $0.007689$.} However, the three clusters have different centrality with $B_3$ being the most central and $B_1$ the least.

Figure~\ref{fig:9_clusters} shows the average scores given to each of $B_1$, $B_2$, $B_3$ by \pr{} and \fpr{}. The \pr{} scores confirm the difference in centrality between the three clusters. More importantly, \fpr{} also shows an increase in scores from $B_1$ to $B_3$, similar to \pr{}. This shows that \fpr{} reflects the differences in centrality of nodes when both local and neighborhood polarities are held largely constant, giving important nodes a higher score.


\section{Real-world Graphs} 
\label{sec:real}
In this section we apply the \fpr{} to some graphs generated from real-world networks, to show its performance on these datasets and compare the results with conventional wisdom regarding the networks. We choose two real-world graphs where centrality and polarization are relevant: the network of members of the United States Congress and a network of political blogs. 

Before presenting our results, we first present a natural generalization of a different and equally classic centrality notion, \bc~\cite{betweenness}, to capture diversity.  Subsequently, we will evaluate the performance of \fpr{} using \dbc{} along with \pr{} as benchmarks.

\subsection{\dbc{}}
\label{sec:dbc}
The \bc{} of a vertex $v$ quantifies the frequency it appears on the shortest path of two other vertices. Formally,
\[
\bcs(v) := \sum_{s\neq v\neq t} \delta_{st}(v)
\]
where $\delta_{st}(v)$ is the fraction of $s$-$t$ shortest paths passing through $v$.

We naturally generalize it into \dbc{} to incorporate diversity considerations. We define
\[
\dbcs(v) := \sum_{s\neq v\neq t} \delta_{st}(v) |r_s - r_t|
\]
to measure how often it appears on the shortest path of two other vertices, while giving larger weights if those two vertices are from different communities. 

\dbc{} can be computed in $O(nm)$ time on unweighted graphs, where $n$ and $m$ are the numbers of nodes and edges respectively, using a simple modification from Brandes' algorithm \cite{Brandes01}. In practice, it is typically slower than \fpr{} which takes $O(m)$ per iteration.

\subsection{Congress Graph}

\begin{figure}
    \centering
    \begin{subfigure}{0.32\columnwidth}
        \centering
        \includegraphics[width=\columnwidth]{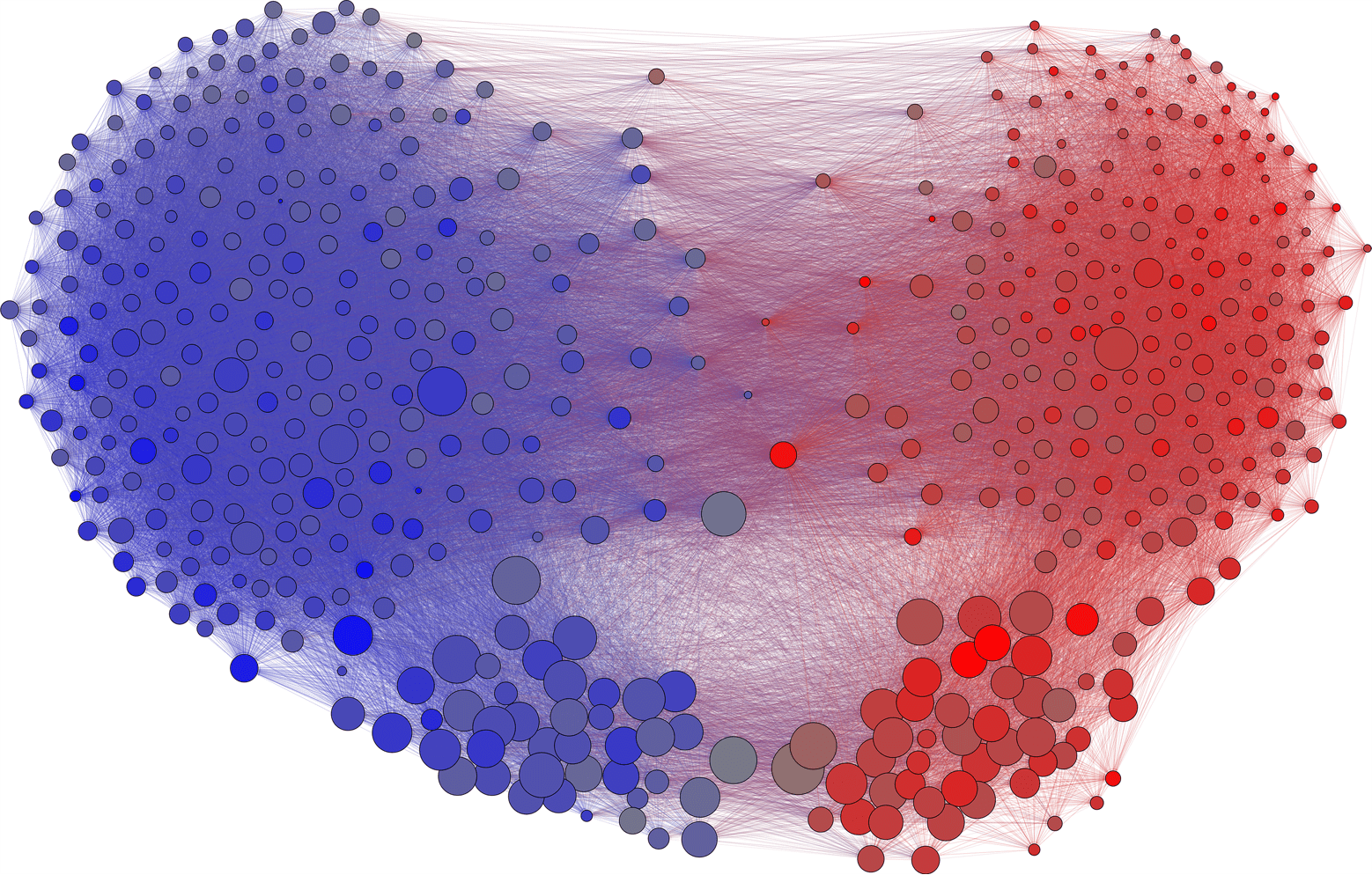}
        \caption{\fpr{}} \label{fig:congress_fair}
    \end{subfigure}
    \begin{subfigure}{0.32\columnwidth}
        \centering
        \includegraphics[width=\columnwidth]{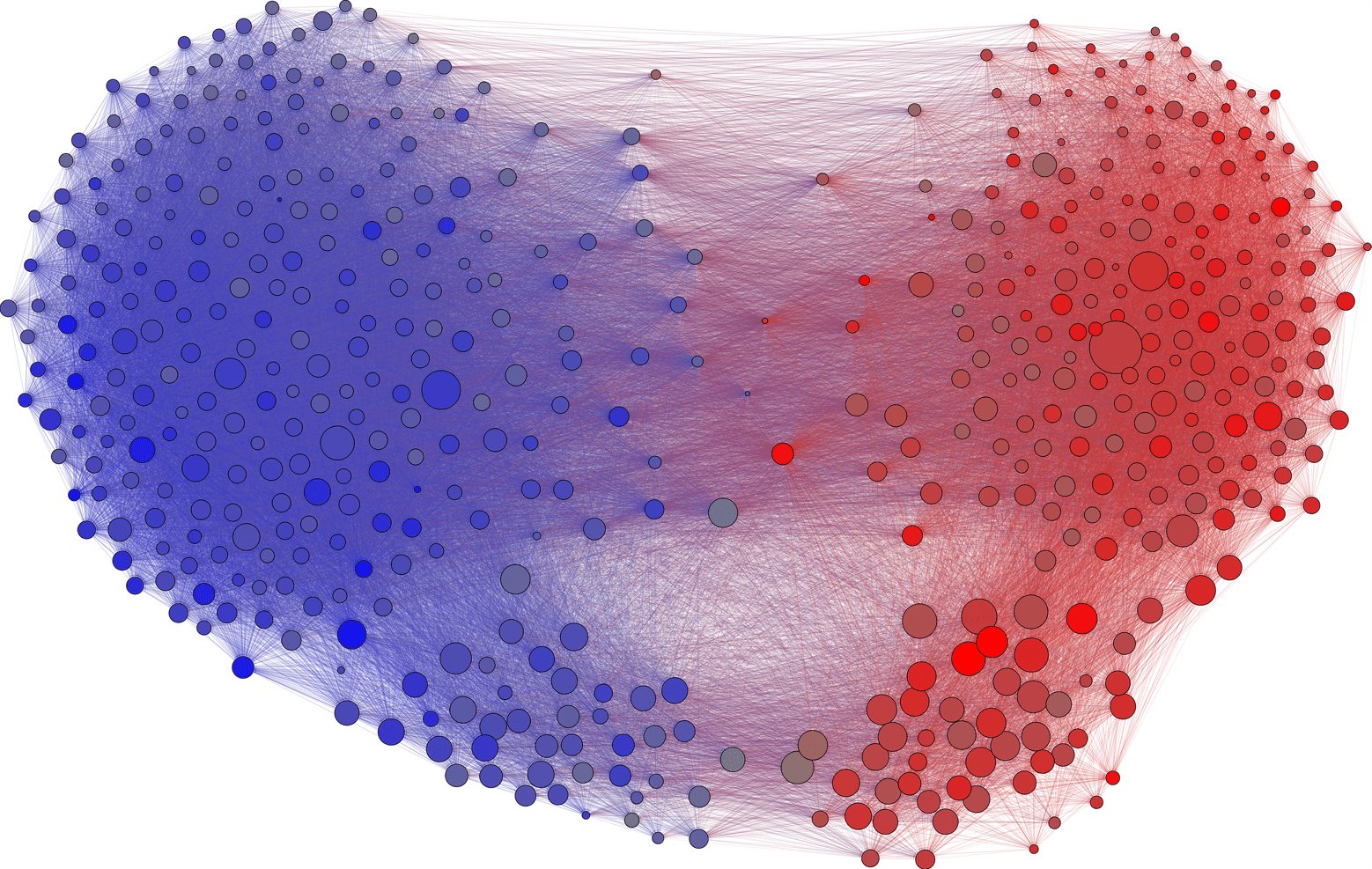}
        \caption{\pr{}} \label{fig:congress_standard}
    \end{subfigure}
    \begin{subfigure}{0.32\columnwidth}
        \centering
        \includegraphics[width=\columnwidth]{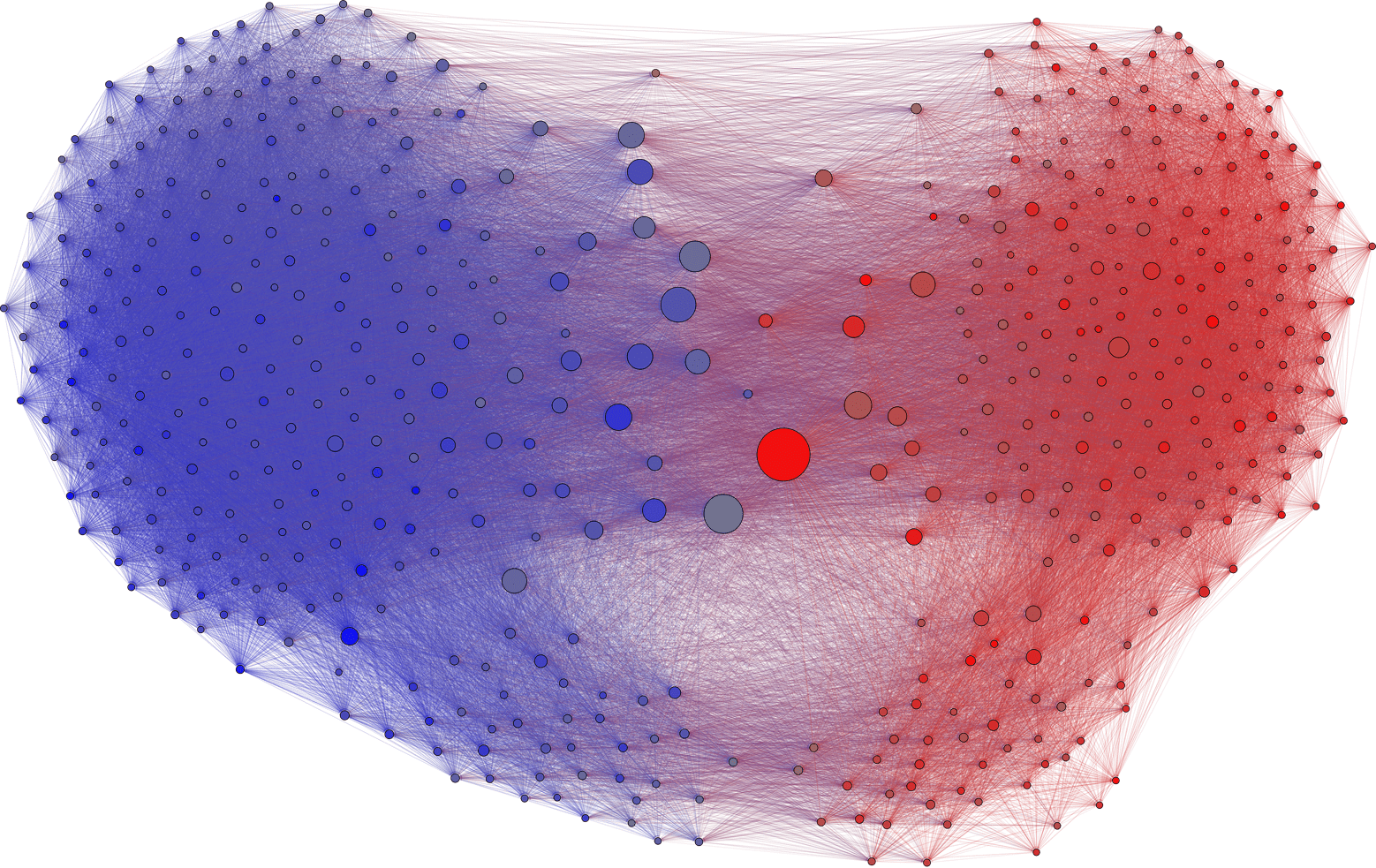}
        \caption{\dbc{}} \label{fig:congress_betweenness}
    \end{subfigure}
    \caption{Congress Graph. Layout generated by the ForceAtlas2 algorithm~\cite{FA2}. Nodes are colored based on polarities ($\mathbf{q}^{(i)}$).}
    \label{fig:congress}
\end{figure}

We create a graph that models all members of the $116^{\text{th}}$ United States Congress (2019--20) linked by their connections on Twitter. Each Congress member with an official Twitter account is represented as a node. A directed edge from member X to member Y exists if X was following Y on Twitter as of May 30, 2020.\footnote{Two members in the graph did not follow any other members. In order to prevent their \pr{} values from being artificially inflated, edges are added from each of the two sinks to all nodes in the graph.}

There are $K=2$ communities, and for node $i$, the affiliation vector $\mathbf{q}^{(i)}=\left[ b_i, r_i \right]$ measures how much the Congress member's ideologies align with the Democrats and Republicans respectively. To generate this vector, we use the DW-NOMINATE procedure that provides a quantitative measure of their ideologies based on factors such as voting behavior~\cite{DWN}. For each member, this method produces a number between $-1$ (most liberal) and $1$ (most conservative), which we scale linearly to fit our notion in Section~\ref{sec:model}.

When comparing \fpr{} with \pr{} and \dbc{} on the Congress Graph, as shown in Figure~\ref{fig:congress}, 
we observe a \emph{bridging effect}:
\begin{quote}
The \fpr{} model favors nodes that are not only bipartisan and influential (in terms of \pr{}), but that also serve as a bridge to similarly bipartisan and influential nodes in other communities.
\end{quote}

Specifically, while the graph in Figure~\ref{fig:congress} is clearly separated into two polarized clusters, each of them contain a sub-cluster shown on the bottom left and bottom right. The sub-clusters highly correlate to the Senate, with the main cluster largely being the House. In general, senators from both parties are well-connected to each other; each senator is connected to a large number of representatives from their own party and a smaller number from the other party; edges directly between representatives of different parties exist, but are less common. This means the Senate sub-clusters serve as bridging nodes between the two clusters, since two representatives from different parties are typically connected via one or more senators.

While both Figures~\ref{fig:congress_fair}~and~\ref{fig:congress_standard} show the Senate sub-clusters ranked above average, the difference between the Senate and the House shown by \fpr{} in Figure~\ref{fig:congress_fair} is much more pronounced, as both sub-clusters receive much higher scores here. Figure~\ref{fig:congress_betweenness}, on the other hand, shows \dbc{} ranks the Senate similarly to the House. Instead, the measure appears to favor nodes with high outdegrees, as the highly ranked members follow most or all other members of Congress on Twitter. While these nodes can serve as a bridge across the two communities, they are not necessarily bipartesan and influential themselves.

\begin{table}[htbp]
    \centering
    \begin{tabular}{|c|c c c|} 
        \hline
        \multicolumn{4}{|c|}{Senators} \\
        \hline
        Node & \fprs{} & \prs{} & \dbcs{} \\
        \hline\hline
        Susan M. Collins  & $1$ & $10$ & $152$  \\
        Lisa Murkowski  & $6$ & $20$ & $224$  \\
        Kyrsten Sinema  & $9$ & $23$ & $2$  \\
        Joe Manchin, III  & $5$ & $64$ & $198$  \\
        Rob Portman  & $11$ & $8$ & $30$  \\
        \hline
    \end{tabular}
    \caption{Top 5 senators 
   joining greatest proportion of bipartisan bills, with their ranks by \fpr{}, \pr{}, and \dbc{} respectively.}
    \label{table:congress_bipartisan}
\end{table}

To verify that the \fpr{} indeed finds the most important bipartisan members of the Congress, we consider the top five senators 
with the highest percentage of bills cosponsored that were introduced by the other party, an independent metric that serves as a proxy of bipartisanship of members~\cite{bip}. As shown in Table~\ref{table:congress_bipartisan}, 
4 out of these 5 senators are ranked significantly higher by \fpr{} (\fprs{}) than either \pr{} (\prs{}) or \dbc{} (\dbcs{}). 
Furthermore, these five senators are among top 11 in the entire graph ranked by \fpr{}; 
This indeed shows that \fpr{} prioritizes bipartisan and influential nodes in the graph. 

We can quantify this bridging effect by analyzing intersections of top nodes and the cut between two clusters. Spectral clustering is first performed on the graph (converted to an undirected graph) to identify the two polarized clusters~\cite{spectral}. Then, for each ranking algorithm and for various values of $k$, we count the number of edges in the subgraph induced by the top $k$ nodes that are also in the cut. 
Figure~\ref{fig:congress_spectral_cut} shows the results for all these centrality models.

\begin{figure}
    \centering
    \includegraphics[width=0.85\columnwidth]{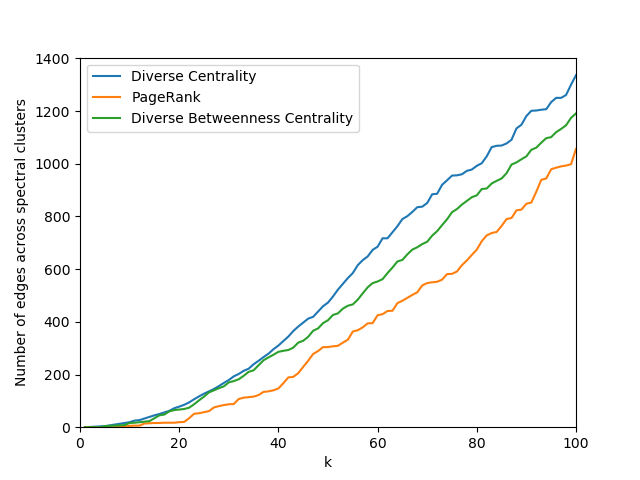}
    \caption{Number of cut edges among top $k$ nodes in Congress Graph ranked by \fpr{}, \pr{}, and \dbc{} scores. }
    \label{fig:congress_spectral_cut}
\end{figure}

The \fpr{} model has a greater number of cut edges for any given value of $k$. This shows that \fpr{} has a greater focus than \pr{} on bridging nodes that connect the two polarized clusters, instead of central nodes within each community which may not be well connected to each other. \fpr{} beats \dbc{} in this metric too, even though the latter emphasizes this bridging effect in its definition.

\subsection{Political Blogs Graph}

We also ran the algorithms on a graph that models blogs on U.S. politics in 2005~\cite{blog}. The dataset is a directed graph that forms edges among blogs based on hyperlinks.
\footnote{Four nodes are not connected the rest of the graph, so we removed them and other disconnected nodes to ensure spectral clustering finds a good cut.} 

We again assume $K=2$ communities in the graph. Since the dataset does not contain fine-grained polarity data of the blogs, for our adaptation, all liberal blogs are assigned $r_i = 0.01, b_i = 0.99$, and all conservative blogs have $r_i = 0.99, b_i = 0.01$. These values are approximately discrete, denoting community memberships rather than extent of affiliation to each community.


We can again analyze the number of cut edges among top nodes to quantify this bridging effect, as shown in Figure~\ref{fig:political_blogs_spectral_cut}. For any given value of $k$, \fpr{} gives the greatest number of cut edges across the two clusters than \pr{}, similar to the Congress Graph. Therefore, on this graph, \fpr{} still favors nodes on the frontier of the two clusters bridging them together. This example also shows that \fpr{} still performs well when nodes have extreme affiliation vectors (in the sense of being nearly 0 or 1).

\begin{figure}
    \centering
    \includegraphics[width=0.85\columnwidth]{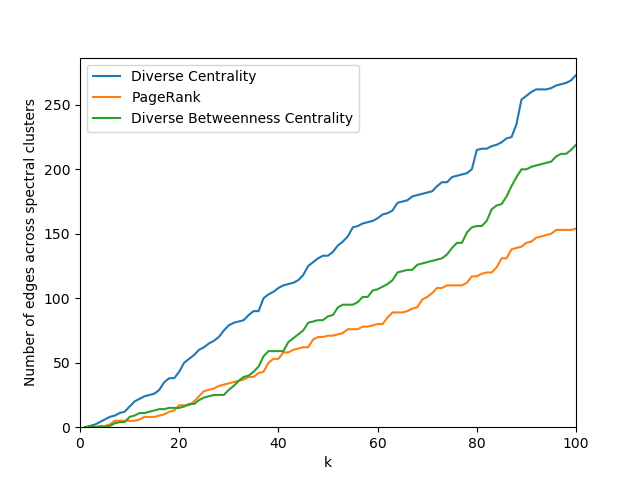}
    \caption{Number of cut edges between top $k$ nodes in Political Blogs Graph ranked by \fpr{}, \pr{}, and \dbc{} scores.}
    \label{fig:political_blogs_spectral_cut}
\end{figure}

\section{Conclusion}
In this paper, we defined a new measure of centrality, \fpr{}, that gives importance to nodes that bridge different communities. We propose several elementary criteria for such a centrality measure with diversity and show via simulations on realistic random networks that \fpr{} satisfies all the proposed criteria, while simple modifications to standard \pr{} do not.  We further show how important bipartisan senators are naturally unearthed by our measure when run on a social network constructed from the United States Congress. 

Our work leads to several open questions. First, we do not have a complete theoretical understanding of \fpr{}. Is the fixed point unique under reasonable assumptions on the network? Can it be proven that our iterative algorithm for computing the fixed point always converges? Further, in several real-world settings, there could be several different types of communities, corresponding for instance to politics, demographics, religion, location, etc. An individual could belong to one or more communities for each of these dimensions. How can we generalize the notion of \fpr{} to handle this setting, and how does it impact the performance of the resulting algorithms on real data?

\section*{Acknowledgments}
This work is supported by NSF grants CCF-1637397, ONR award N00014-19-1-2268; and DARPA award FA8650-18-C-7880.

\bibliographystyle{abbrv}
\bibliography{refs.bib}

\end{document}